\documentclass[11pt,twoside]{article} 
\usepackage{ifthen}
\usepackage{soul}
\usepackage{fullpage, prettyref}

\usepackage{fancyhdr}
\usepackage{times}
\usepackage{graphicx}
\usepackage{marvosym}
\usepackage{ifthen}
\usepackage{algorithm}
\usepackage{thmtools}

\usepackage[noend]{algpseudocode}
\algrenewcomment[1]{\(\triangleright\) {\small{\emph{\color{blue} #1}}}}
\algnewcommand{\LineComment}[1]{\State \(\triangleright\) \emph{\color{blue} #1}}
\algnewcommand{\Invariant}[1]{\State \(\triangleright\) \emph{\color{red} #1}}

\usepackage{amsmath, amssymb, amsthm}
\usepackage{url}
\usepackage{boxedminipage}
\usepackage{wrapfig}
\usepackage{ifthen}
\usepackage{color,xcolor}
\usepackage{transparent}

\usepackage{varwidth}

\usepackage{a4,geometry}
\usepackage{graphicx}
\usepackage{amsmath,amssymb,amsthm,mathtools}
\usepackage{bm}
\usepackage{xspace}
\usepackage{url}
\usepackage{boxedminipage}
\usepackage{wrapfig}
\usepackage{ifthen}
\usepackage{color}
\usepackage{xcolor}
\usepackage{framed}
\usepackage{mdframed}
\usepackage[colorlinks=true,pdfpagemode=none,urlcolor=blue,linkcolor=blue,citecolor=violet,pdfstartview=FitH]{hyperref}
\usepackage[noabbrev,nameinlink]{cleveref}
\usepackage{fullpage}
\usepackage{enumitem}
\usepackage{tikz}

\newtheorem{theorem}{Theorem}
\newtheorem{lemma}{Lemma}
\newtheorem{claim}{Claim}

\newtheorem{definition}{Definition}

\newcommand{\ignore}[1]{}

\newcommand{\calJ}{\mathcal{J}}

\newcommand{\calP}{\mathcal{P}}

\newcommand{\bx}{\mathbf{x}}

\newcommand{\ceil}[1]{{\left\lceil{#1}\right\rceil}}
\newcommand{\floor}[1]{{\left\lfloor{#1}\right\rfloor}}

\def\2plus{{\tt (++)}}
\def\3plus{{\tt (+++)}}
\def\4plus{{\tt (++++)}}
\def\5plus{{\tt (+++++)}}

{\hspace*{\fill}$\Box$\par}
\newlength{\algobox}
\colorlet{shadecolor}{blue!10}

\def\eq{\leftarrow}

\def\SUM{\mathsf{sum}}

\def\CUT{\mathsf{CUT}}

\def\RANK{\mathsf{rank}}
\def\ADDITIVE{\mathsf{add}}
\def\RANKP{\RANK}

\def\com{\mathsf{com}}

\def\calP{\mathcal{P}}

\def\rep{\mathsf{rep}}

\usepackage{times}

\usepackage{paralist}

\title{Learning Partitions using Rank Queries\footnote{The work was supported by NSF grants CCF-2041920, 2402571}}
\author{Deeparnab Chakrabarty \and Hang Liao}
\date{}

\begin{document}
	
	\maketitle
	
	\begin{abstract}%
	
	\noindent
	We consider the problem of learning an unknown partition of an $n$ element universe using rank queries. Such queries take as input a subset of the universe and return the number of parts of the partition 
	it intersects. We give a simple $O(n)$-query, efficient, deterministic algorithm for this problem. 
	We also generalize to give an $O(n + k\log r)$-rank query algorithm for a general partition matroid where $k$ is the number of parts and $r$ is the rank of the matroid.
\end{abstract}

\section{Introduction}

Let $V$ be a universe of $n$ elements and suppose there is an {\em unknown} partition $\calP = (P_1, \ldots, P_k)$ that we want to learn.
We have an oracle called $\RANK$ that takes as input any subset $S\subseteq V$ and returns the number of different parts this subset intersects. More precisely
$\RANK(S) := \sum_{i=1}^k \min(|S\cap P_i|, 1)$. How many queries suffice to learn $\calP$? 

This natural question 
is a special case of the problem of {\em learning hypergraphs} under the {\em additive query} model initially studied by~\cite{GrebinskiK00}. In this problem, we have an unknown hypergraph on a vertex set $V$, and an additive query $\ADDITIVE(T)$ on a subset $T\subseteq V$ returns 
the {\em number} of hyperedges completely contained in $T$. Our unknown partition $\calP$ is a special hypergraph whose $k$ hyperedges are disjoint (that is, it is a hypermatching);
and for any subset $S$ we observe that $\RANK(S)$ is precisely $k - \ADDITIVE(V\setminus S)$. And so, the problem we study can be rephrased as in how few additive queries can a hypermatching be learnt.
Although hypermatchings may feel too specialized, the now mature literature on {\em graph} learning (cf.~\cite{Choi13,BshoutyM12,BshoutyM11,Bshouty2011parity,ChoiK10}) began with understanding the case of graph matchings (cf.~\cite{GrebinskiK00,AlonBKRS02,AlonA05}). 

The problem we study is also a special case of a {\em matroid learning} problem with access to rank oracle queries. Matroids are set systems, whose elements are called {\em independent} sets, that are defined using certain axioms and these are fundamental objects in combinatorial optimization.
It is well known that a partition $\calP$ induces the following simple partition matroid: a subset $I\subseteq V$ is independent if $|I\cap P_i| \leq 1$ for all $i$. 
The rank of a matroid is the cardinality of the largest independent set of the matroid, and more generally, the rank of subset $S$ is the cardinality of the largest independent set that is a subset of $S$. A moment's notice shows that for the simple partition matroid this is precisely $\RANK(S)$ which explains the name we give to our oracle. So, our problem we study asks: in how few rank queries can a simple partition matroid be learnt?\smallskip

It is rather straightforward\footnote{Something that can be given in an undergraduate algorithms course when teaching binary search.} to learn the partition using $O(n\log k)$ queries as follows. First, one learns a representative from each part with $n$-queries; given 
a set of already learned representatives $R$, a vertex $v$ is in a new unrepresented part if and only if $\RANK(R\cup v) > \RANK(R)$. After learning the $k$ representatives, we can learn every other vertex's part by performing a binary search style algorithm. Can one do better?
It is instructive to note that the algorithm sketched above does not really utilize the full power of the query model we have. In particular, it would have sufficed if the query took a subset $S$ and said YES if every element in $S$ was in a different part, or NO otherwise. Using the matroid language, an {\em independence oracle} suffices which only states if a set $S$ is independent or not. 
Now, an independence oracle answer gives at most $1$ bit of information; on the other hand, there are roughly $k^n$ different partitions possible with $\leq k$ parts. Therefore, via an information theoretic argument $\Omega(n\log k)$ independence queries are {\em necessary} to learn the partition. In contrast, the {\em rank} oracle gives the {\em number} of different parts hit by a subset; this is an integer in $\{0,1,\ldots, k\}$ and the information theoretic argument only proves an $\Omega(n)$ lower bound on the number of queries. This naturally leads to the question: can an $O(n)$-query algorithm exist? The main result of this paper is a simple affirmative answer to this question.

\begin{theorem} \label{thm:partition}
	There is a deterministic, constructive algorithm that solves unknown partition learning problem using $O(n)$ $\RANK$ queries. 
\end{theorem}

\noindent
{\em Remark.}
\emph{We have not optimized the constant in front of $n$. We think it can be made less than $10$ but don't believe can be made less than $4$ using our methods. The best lower bound one can prove using the above information theory argument is $n$.
	Figuring out the precise coefficient is left as an open question.}	\medskip

\noindent
We also consider the generalization of learning a {\em general} partition matroid using rank queries. In this case, along with the unknown partition $\calP$, we have unknown positive integers
$r_1, \ldots, r_k$ associated with each part, where $1\leq r_i < |P_i|$. A subset $I$ is independent in this matroid if $|I\cap P_i| \leq r_i$, for all $1\leq i\leq k$. When all $r_i=1$, we have the simple partition matroid. The rank query corresponds to $\RANK(S) := \sum_{i=1}^k \min(|S\cap P_i|, r_i)$. In how few rank queries can we learn a general partition matroid? 

As in the simple partition matroid case, one can get an $O(n\log k)$-query algorithm using just an independent set oracle
via a more delicate\footnote{Maybe a challenging exercise in the aforementioned algorithms course; see~\Cref{sec:genpart} for this algorithm.} binary-search-style algorithm.  Can we obtain $O(n)$ query algorithm with rank queries? We believe the answer should be yes and take the following first step.

\begin{restatable}{theorem}{matroid} \label{thm:matroid}
	There is a deterministic, constructive algorithm that learns a general partition matroid using $O(n + k\log r)$ $\RANKP$ queries where $r := \RANK(V) = \sum_{i}r_i$. 
\end{restatable}

\noindent
\emph{Remark. When the number of parts $k\leq n/\log n$, we thus get an $O(n)$-rank query algorithm. However, when $k = \Omega(n)$ we don't do any better than just with independence queries. 
}

\paragraph*{Perspective.} Our motivation to look at the problem arose from trying to understand the {\em connectivity} question in hypergraphs using $\CUT$ queries.
Although, as mentioned earlier, graph learning under query models has been extensively studied, over the last few years, multiple works such as~\cite{RubinsteinSW18,Graur,LeeMS21,AuzaL21,Apers,ChakrabartyL23,Liao24a} have focused on trying to understand if fewer queries can lead to understanding {\em properties} of graphs.
Of particular interest is understanding the {\em connectivity}/finding spanning forest of a graph using $\CUT$ queries. 
A $\CUT$ query takes a subset of vertices as input and returns the number/weight of the cut edges crossing the subset.
While graph learning can take $\tilde{\Theta}(m)$ cut queries, a spanning forest of an undirected graph, unweighted or weighted, can be constructed\footnote{using a randomized Las Vegas algorithm which makes $O(n)$ queries in expectation} in $O(n)$ queries (see~\cite{Apers,Liao24a}). Can such results be generalized\footnote{At first glance even a polynomial query algorithm may not be clear; a little thought can lead to an $O(n\log n)$ query algorithm.} to hypergraphs? To us, the easiest case of a hypergraph was the hypermatching whose only spanning forest is the hypergraph itself. It is not too hard to see that $\CUT$ queries and $\RANK$ queries are intimately related. Formally, after $n$ cut queries, any rank query can be simulated with $2$ cut queries.
The interesting open question is: {\em can the connectivity question of an arbitrary hypergraph be solved in $O(n)$ queries?}
%
%

The other related problem is {\em matroid intersection}. Given rank/independence oracle to two matroids over the same universe, the matroid intersection problem asks to find 
the largest common independent set. It is a classic result in combinatorial optimization due to~\cite{Edmonds72} that this can be solved in polynomially many independence oracle queries.
The current state-of-the-art is that $\tilde{O}(n^{1.5})$-rank queries suffice (see~\cite{ChakLSSW19}) and $\tilde{O}(n^{7/4})$-independence oracle queries suffice (see~\cite{Bilkstad21}). 
On the other hand, no {\em super-linear} lower bounds are known for rank-queries, and only recently, \cite{BilkMNT23} proved an $\Omega(n\log n)$-lower bound for independence queries.
The big open question is: {\em can matroid intersection be solved in $O(n)$ rank-queries, or can a $\omega(n)$-lower bound be proved}? 

As noted earlier, if we wish to obtain an $o(n\log k)$-query algorithm, we must exploit the fact that $\RANK$-queries output ``more'' than the independence oracle queries. 
Our second motivation in writing this paper is to showcase how the techniques that arise from {\em coin weighing} problems a la \cite{Cantor,Lindstrom} exploit this ``more''.
In the basic coin-weighing problem, one is asked to recover an unknown Boolean vector $x$ with the ability to query any subset $S$ and obtain $\sum_{i\in S}x_i$ (a sum-query). The aforementioned papers showed how to do this making roughly $2n/\log_2 n$ sum-queries. 
\cite{Bshouty09} generalized this to learn a Boolean vector with at most $d$ ones in roughly $\frac{2d\log_2 n}{\log_2 d}$ queries. In a different application, \cite{GrebinskiK00} showed how to use the coin-weighing result to learn a hidden perfect matching in a bipartite graph using $2n$ $\CUT$ queries. These form the backbone of our algorithms.
Having said that, there are some big differences between sum-queries and rank-queries since the latter is not ``linear'' and 
this underlies the difficulties we've faced in generalizing~\Cref{thm:matroid} to obtain a $O(n)$-query algorithm to the general partition matroid case.

\subsection{Related works}

There is a vast literature on combinatorial search~\cite{Aigner88,DuH00}, and we restrict ourselves to the works that are related the most.
As mentioned above, our problem can be thought of as learning a {\em hypermatching} using additive/cut queries. (Hyper)-graph reconstruction questions have been widely studied in the last two decades.
A significant body of work has been dedicated to reconstructing graphs using queries, as evidenced by the works of (cf. \cite{GrebinskiK00,AlonBKRS02,AlonA05,ReyzinS07,ChoiK10,Mazzawi10,BshoutyM11,BshoutyM12,Choi13}). 
These efforts encompass various types of graphs, including unweighted graphs, graphs with positive weights, and graphs with non-zero edge weights, using $\CUT$ queries.
This has culminated in a result of~\cite{Choi13} gives an polynomial time, randomized $O(\frac{m\log n}{\log m})$-query algorithm for learning graphs on $n$ nodes and $m$ edges with non-zero edge weights, and this query complexity is information theoretic optimal. 
Concurrently, there has been ongoing research on recovering specific structures within graphs without necessarily reconstructing the entire graph, such as figuring its connectivity (see~\cite{GrebinskiK00, ChakrabartyL23, Liao24a}).  \cite{Angluin05} started the research on learning a hypergraph using edge-detecting queries, that is, whether the input set contains a hyperedge or not; they described algorithms for $r$-uniform hypergraphs (every hyperedge has exactly $r$ vertices) but the dependence on $r$ was exponential. \cite{BshoutyM10} considered the {\em additive} model where one gets the number of edge (this was mentioned in the Introduction above) and proved existence of algorithms to 
learn rank $d$ hypergraphs (every hyperedge has at most $d$ vertices) for constant $d$ using $O_d(m\log(n^d/m)/\log m)$-queries; the dependence on $d$ is exponential.
%
\cite{Balkanski22} considered high-rank but low-degree hypergraphs, including hypermatchings. The focus was on edge-detecting queries (indicator whether additive query is zero or non-zero), and they gave $O(n\mathrm{polylog} n)$-query algorithms which were also ``low depth'', that is, with few rounds of adaptivity. 

Our problem is also related to the problem of recovering a {\em clustering} with active queries (see~\cite{AshKB16,MazS17,SahS19,AilBJ18,ChiZL19,BreCLP21,Liu22}). The setting is the same: there is a universe of $n$ points which we assume is clustered into $k$ unknown parts. The query model, however, is often quite different and much more restrictive usually constraining queries to asking whether a pair or a constant number of elements are in the same cluster/part or not. Such a study was initiated in the works of~\cite{DavKMR14,AshKB16,MazS17} which prove an $\Omega(nk)$ lower bound, and then provide better upper bounds with extra assumptions. The above-cited works continue on this line. 

\section{$O(n)$ Query Deterministic Algorithm}\label{sec:simpart}

Throughout the rest of the paper, unless otherwise mentioned all logarithm base is $2$.
We begin with an overview of the algorithm. 
We maintain a collection $\calJ$ of disjoint independent sets; recall that a subset is independent if it contains at most one element from each part. Initially, $\calJ$ is the collection of $n$ independent sets each of which is a single element.
Let $J$ denote the union of all these independent sets, and so, initially $J = V$. 
The algorithm will modify this collection $\calJ$ in iterations, removing some elements from $J$ while doing so. Anytime such an element $e$ is removed, we maintain a map $\rep(e)$ to an element in the current $J$
with the property that $e$ and $\rep(e)$ are in the same partition in $\calP$. We will call two elements in the same parts ``friends'', and so, $\rep(e)$ is $e$'s friend.

The key routine in the algorithm is a merge operation over independent sets. 
Given two independent sets $I_1, I_2$, define the set of {\em common nodes} $\com(I_1, I_2) := \{v_1 \in I_1: \exists v_2 \in I_2, P_i, \{v_1,v_2\} \subseteq P_i\}$
to be the subset of nodes in $I_1$ which have a friend in $I_2$.
Note that this friend needs to be unique since $I_2$ is independent. The {\sc Merge} operation takes two independent sets $I_1$ and $I_2$
and then (a) finds the set $\com(I_1, I_2)$, (b) for each $e\in \com(I_1, I_2)$, finds its unique neighbor $\rep(e) \in \com(I_2, I_1)$, and (c) returns $\com(I_1, I_2), \com(I_2, I_1)$, and $I_3 := I_1 + I_2 - \com(I_1, I_2)$. See \Cref{fig1} for an illustration. 

\begin{figure}[h]
    \centering
    \includegraphics[width=0.75\textwidth]{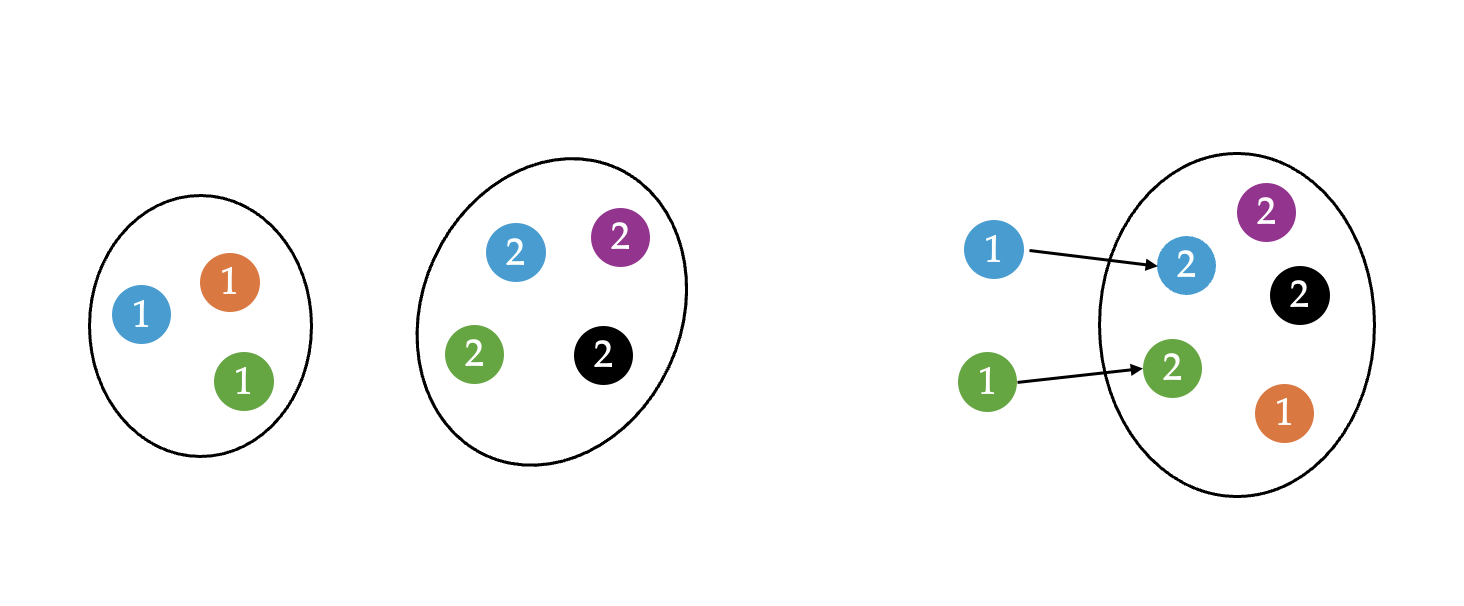}
    \caption{After we merge $I_1, I_2$ on the left, we get $I_3$ and a mapping from $\com(I_1, I_2)$ to $\com(I_2, I_1)$.}
    \label{fig1}
\end{figure}

Given the {\sc Merge} routine, the algorithm is very simple: while there exists two independent sets $I_1$ and $I_2$ of comparable size (within factor $2$), merge them
and replace $I_1$ and $I_2$ with $I_3$ returned by the merge. This may remove some elements from $J$, and in particular this is $\com(I_1,I_2)$, but all these elements
will have $\rep(e)$ pointing to their friends who are still in $J$. When the algorithm can't do this anymore, there must be at most $\ell = \ceil{\log k}$ independent sets remaining in $\calJ$.
These can be sequentially merged in any order to get one single independent set $J$, indeed a basis, in $\calJ$. To find the partition, consider the directed graph on $V$ where 
we add the edge $(e,\rep(e))$ for all $e\in V\setminus J$; note this forms a collection of directed in-trees rooted at vertices in $J$, and the connected components are precisely the parts that we 
desire. 
In what follows we show how to implement {\sc Merge} using existing results from coin-weighing and graph reconstruction, and then argue why the total number of $\RANK$ queries made by our algorithm is $O(n)$.

\subsection{Definitions}

We review or introduce several definitions for completeness.

\begin{definition} ($\ADDITIVE$ query)
An $\ADDITIVE$ query on an unweighted graph $G = (V, E)$: given $S \subseteq V$, obtain $|\{e \in E: e \in S \times S\}|$.
\end{definition}

\begin{definition} ($\SUM$ query)
A $\SUM$ query on a boolean vector $x \in \{0,1\}^N$: given $S \subseteq [N]$, obtain $\sum_{i\in S} x_i$.
\end{definition}

\begin{definition} ($\com$ of two sets)
Given two independent sets $I_1, I_2$. The set of {\em common nodes} $\com(I_1, I_2) := \{v_1 \in I_1: \exists v_2 \in I_2, P_i, \{v_1,v_2\} \subseteq P_i\}$
is the subset of nodes in $I_1$ which have a friend in $I_2$.
\end{definition}

\begin{definition} ($\rep(e)$ of a node)
We maintain a map $\rep$ with the property that a node $e$ and $\rep(e)$ (representative of $e$) are in the same partition in $\calP$. $\rep$ keeps track of the learned partition by mapping the learned node to its friend who is still in $J$. 
\end{definition}

\subsection{Merging Independent Sets.}
We begin by introducing some vector/graph reconstruction algorithms from the literature.

\begin{lemma}[\cite{Bshouty09}]
	\label{lem:bshouty}
	Let $x\in \{0,1\}^N$ be an unknown Boolean vector with $\SUM$-query access. If $x$ has $d$ ones, then there is a polynomial time, adaptive, deterministic algorithm to reconstruct $x$ which makes $O(d\log(N/d)/\log d)$ $\SUM$ queries.
\end{lemma}

\begin{lemma}~[Paraphrasing Theorem 4 \& Section 4.3 \cite{GrebinskiK00}] \label{lem:GK}
	A bipartite graph $G = (V,W,E)$ with $|V| = |W| = m$ where $E$ forms a perfect matching can be learnt in 
		$O(m)$ $\ADDITIVE$ queries. 
\end{lemma}
\noindent
\noindent
Now we are ready to describe {\sc Merge} whose properties are encapsulated in the following lemma.
\begin{lemma} \label{lem:merge}
	Let $I_1, I_2$ be two independent sets and let $k_1 = |I_1|$ and $k_2 = |I_2|$. Suppose $d = |\com(I_1, I_2)| = |\com(I_2, I_1)|$. 
	The procedure {\sc Merge} is an adaptive deterministic polynomial time algorithm which returns $I_3 = I_1 + I_2 - \com(I_1,I_2)$ 
	and $\rep(e) \in \com(I_2,I_1)$ for all $e\in \com(I_1,I_2)$. The procedure makes $O(\frac{d\log (
	\max(k_1,k_2)/d)}{\log d})$ $\RANK$ queries. 	
\end{lemma}
\begin{proof}
Given $I_1$ and $I_2$, define the Boolean vector $\bx := \bx_{(I_1,I_2)} \in \{0,1\}^{k_1}$ where $\bx_e = 1$ if and only if $e \in \com(I_1,I_2)$.
We note that a $\SUM$ query can be simulated on $\bx$ using a single $\RANK$ query. This is due to the observation that for all $S\subseteq I_1$, 
$\sum_{e\in S} \bx_e = |S| + |I_2| - \RANK(S \cup I_2)$.
This is because the RHS precisely counts the number of parts of $S$ that are already present in $I_2$, or $\com(I_1,I_2) \cap S$. 
Therefore, 
we can apply~\Cref{lem:bshouty} to learn $\com(I_1,I_2)$ in $O(\frac{d\log (k_1/d)}{\log d})$ many $\RANK$ queries.
Similarly, we can get $\com(I_2,I_1)$ in $O(\frac{d\log (k_2/d)}{\log d})$ queries. 
Note that the above doesn't give us the friends for $e\in \com(I_1,I_2)$ in $\com(I_2,I_1)$. This pairing can be found as follows.
For simplicity, let's use $X := \com(I_1,I_2)$ and $Y := \com(I_2,I_1)$. Consider the bipartite graph $G = (X,Y,E)$ where
$e\in X$ has an edge to $f\in Y$ if and only if $f$ is $e$'s friend. So, $G$ is a perfect matching whose edges are yet unknown.
We can now use~\Cref{lem:GK} to find them.  To see why this can be done, note that we can simulate the $\ADDITIVE$ query because for any $S \subseteq X \cup Y$, simply 
because
$\ADDITIVE(S) = |S| - \RANK(S)$
This is because any edge $(e,f)$ with both endpoints in $S$ are precisely the pairs which are counted once in $\RANK(S)$ but twice in $|S|$.
Thus, finding this matching takes $O(d)$ $\RANK$ queries. \qedhere

\begin{algorithm}[ht!]
	\caption{Merging Independent Sets}\label{alg:merge}
	%
	\begin{varwidth}{\dimexpr\linewidth-2\fboxsep-2\fboxrule\relax}
		\begin{algorithmic}[1]
			\Procedure{Merge}{$I_1,I_2$}:
			\LineComment{Input: Two independent sets}
			\LineComment{Output: $\com(I_1,I_2)$ and $\rep(e) \in \com(I_2,I_1)$ for $e\in \com(I_1,I_2)$.}
				\State Learn $\com(I_1,I_2)$ and $\com(I_2,I_1)$ as described above using $O(d\log(\max(k_1,k_2))/\log d)$ $\RANK$ queries.
				\State Learn $\rep(e) \in \com(I_2,I_1)$ for $e\in \com(I_1,I_2)$ as described above in $O(d)$ $\RANK$ queries.
				\State $I_3 \eq I_1 \cup I_2 - \com(I_1,I_2)$. \Comment{Note that $I_3$ is independent and $\rep(e) \in I_3$ for all $e\in \com(I_1,I_2)$.}
				\State \Return $(I_3,\rep)$
			\EndProcedure
		\end{algorithmic}
	\end{varwidth}%
\end{algorithm}
%
\end{proof}

\subsection{The algorithm and analysis.} We give the pseudocode of the algorithm 	in~\Cref{alg:main}.
\begin{algorithm}[ht!]
	\caption{Find Partition}\label{alg:main}
	%
	\begin{varwidth}{\dimexpr\linewidth-2\fboxsep-2\fboxrule\relax}
		\begin{algorithmic}[1]
			\Procedure{FindPartition}{$V,\RANK$}:
			\LineComment{Input: $n$ elements with $\RANK$ query access to hidden partition $\calP$.}
			\LineComment{Output: the partition.}			
			\State Create $\calJ \eq \{\{e_1\}, \{e_2\}, \ldots, \{e_n\}\}$; $J\eq V$
			\State Create graph $G = (V,F)$ with $F\eq \emptyset$. \Comment{this will be used to find the parts}
			\While{$\exists I_1,I_2 \in \calJ~:~ |I_1|/|I_2| \in [1/2,2]$}:
				\State $(I_3,\rep(e)) \eq ${\sc Merge}$(I_1,I_2)$. \label{alg:q1}
				\State For all $e\in \com(I_1,I_2)$, add $(e, \rep(e))$ to the edge-set $F$.	
				\State $\calJ \eq \calJ - \{I_1,I_2\} + I_3$; $J \eq J\setminus \com(I_1,I_2)$.
			\EndWhile
			\LineComment{At this point there can be at most $\ceil{\log_2 n}$ elements in $\calJ$}
			\LineComment{Merge all these sets in any order to get a single set. We provide one below.}
			\State Let $\calJ = \{I_1, I_2, \ldots, I_\ell\}$ with $\ell \leq \ceil{\log_2 n}$; $I \eq I_1$; $\calJ \eq \calJ \setminus I_1$
			\For{$2\leq t\leq \ell$}:
				\State $(I_3, \rep(e)) \eq $ {\sc Merge}($I_t,I$). \label{alg:q2}
				\State For all $e\in \com(I_t,I)$, add $(e, \rep(e))$ to the edge-set $F$.	
				\State $\calJ \eq \calJ - \{I_t,I\} + I_3$; $I\eq I_3$.
			\EndFor
			\LineComment{At this point $\calJ$ has a single independent set $I$. Every element in $e \in V\setminus I$ has a single representative $\rep(e)$. So $G$ is a collection of directed in-trees with roots in $I$}
			\State \Return Connected components of $G$.
			\EndProcedure
		\end{algorithmic}
	\end{varwidth}%
\end{algorithm}

\noindent
We now claim that the algorithm makes $O(n)$ queries. All the queries to $\RANK$ occur in the calls to {\sc Merge} in~\Cref{alg:q1} or~\Cref{alg:q2}. 
Let's take care of the second ones first since it's straightforward.
\begin{claim}
	The total number of $\RANK$ queries made in {\sc Merge} calls in~\Cref{alg:q2} over  the for-loop is $O(n)$.
\end{claim}
\begin{proof}
	There are $\ell = O(\log n)$ merges made; that is the only fact we will use. By~\Cref{lem:merge}, the $t$th 
	{\sc Merge} would make at most $O(d_t\log n/\log d_t)$ many $\RANK$ queries, where $d_t$ is the size of $\com(I_t,I)$ at that time.
	All we care for is that $\sum_{t=1}^\ell d_t \leq n$. Now we observe (an explicit reference is Claim 3 of~\cite{ChakrabartyL23}) 
	that if $\ell \leq C\log n$, then $\sum_{t=1}^\ell \frac{d_t}{\log d_t} = O(n/\log n)$. To see this, note that the contribution to this sum of all the 
	$d_t$'s which are $\leq \frac{n}{C\log^2 n}$ is at most $\frac{n\ell}{C\log^2 n} < n/\log n$. 
	All the other $d_t$'s have $\log d_t = \Omega(\log n)$ and so their contribution is $O(\sum_t d_t/\log n) = O(n/\log n)$.
	Altogether, we see that $O(\sum_{t = 1}^\ell d_t\log n/\log d_t) = O(n)$.
\end{proof}

\begin{claim}
	The total number of $\RANK$ queries made in {\sc Merge} calls in~\Cref{alg:q1} over  the while-loop is $O(n)$.
\end{claim}
\begin{proof}
To argue about the {\sc Merge}'s in~\Cref{alg:q1}, we need to partition these into two classes. Note that all such merges take two independent 
sets $I_1$ and $I_2$ which are of similar size $k_1$ and $k_2$ respectively; without loss of generality, let $k_1 \leq k_2 \leq 2k_1$.
Let $d := |\com(I_1,I_2)|$. We call a merge {\em thick} if $d \geq \sqrt{k_1}$
and {\em thin} otherwise. We argue about the thick and thin merges differently. 
\begin{asparaitem}
	\item Using~\Cref{lem:merge}, we see that a thick merge costs $O(d\log(\max(k_1,k_2))/\log d) = O(d)$ $\RANK$ queries; we have used here
	that $k_2 \leq 2k_1$ and $d \geq \sqrt{k_1}$. Thus, we can {\em charge} these $\RANK$ queries to the $d$-elements which {\em leave} $J$.
	Thus, the total number of $\RANK$ queries made across all thick merges is $O(n)$.
	
	\item To argue about thin merges, we make a further definition. Let us say that an independent set $I$ is in class $t$ if $|I| \in [2^t, 2^{t+1})$, for $0\leq t\leq \floor{\log_2 n}$.
	Fix such a $t$.
	A thin merge $(I_1,I_2)$ is called a class $t$ thin-merge if the smaller cardinality set is in class $t$. 
	An element $e\in V$ participates in a class $t$ thin-merge $(I_1,I_2)$ if it is present in the smaller set.
	Observe that for a thin class $t$ merge, the resulting independent set $I_3$ almost doubles in size; in particular, $|I_3| = |I_1| + |I_2| - |\com(I_1,I_2)| \geq 2^{t+1} - 2^{t/2}$.
	Using this one can argue that the same element cannot participate in more than {\em two} class $t$-thin merges; after two merges the set ceases to be class $t$.
	In particular, this means the number of thin class $t$ merges is at most $2\cdot n/2^t$, and each such merge, by~\Cref{lem:merge}, can be done with $O(d\log(2^{t+1})/\log d)$
	many $\RANK$ queries where $d = |\com(I_1,I_2)| < 2^{t/2}$. Since $d/\log d$ is an increasing function of $d$, we conclude that any class $t$ thin-merge takes at most 
	$O(2^{t/2}\log(2^{t+1})/\log(2^{t/2})) = O(2^{t/2})$ many $\RANK$ queries.
	Therefore, the total number of $\RANK$ queries made within thin merges is at most 
		$\sum_{t=0}^{\log_2 n} \frac{2n}{2^t}\cdot O(2^{t/2}) = O(n)$ \qedhere
\end{asparaitem}
\end{proof}
The above two claims imply the proof of~\Cref{thm:partition}.

\section{General Partition Matroids}\label{sec:genpart}

We recall the problem. As before, the universe is $V$ and there is a hidden partition $\calP = (P_1, \ldots, P_k)$.
Furthermore, there are integers $r_1, \ldots, r_k$ where $0 < r_i < |P_i|$.\footnote{Suppose we allow $r_i \ge |P_i|$. Let $M := \{i | r_i \ge |P_i|\}$. Now $\RANK(S) = \sum_{i \in M}|S\cap P_i| + \sum_{i \notin M} \min(|S\cap P_i|, r_i)$. This means we get no information for partitions with index in $M$. To see this, we pick $x_1 \in P_{i_1}, x_2 \in P_{i_2}$ with $i_1, i_2 \in M$ and $i_1 \neq i_2$. If we swap $x_1$ with $x_2$ in every set we give to the $\RANK$-oracle, the answer it returns is the same. Thus no $\RANK$ query algorithm can tell $x_1, x_2$ apart.} This defines a partition matroid where a set $I$ is independent if and only if
$|I\cap P_i| \leq r_i$ for $1\leq i\leq k$. The $\RANK$-oracle for this matroid is the following
$
\RANK(S) = \sum_{i=1}^k \min(|S\cap P_i|, r_i)
$.
\noindent
We will prove the following theorem in this section. 
\matroid*
\noindent
Our proof technique will be a {\em reduction} to the simple partition matroid setting of~\Cref{sec:simpart}. Before we get there, let's first begin with a simple 
well-known observation.
\begin{lemma} \label{lem:basis}
	There is an $O(n)$ $\RANK$ query algorithm that finds a basis $B $ of a partition matroid.
\end{lemma}

\begin{proof}
	This is standard and we give it below for completeness.
\begin{algorithm}[ht!]
		\caption{Finding a Basis Using Rank Queries}\label{alg:basis}
		%
			\begin{varwidth}{\dimexpr\linewidth-2\fboxsep-2\fboxrule\relax}
					\begin{algorithmic}[1]
							\Procedure{FindBasis}{$V,\RANK$}:
							\LineComment{Input: $n$ elements in $V$ with $\RANK$ query access}
							\LineComment{Output: A basis of $V$.}
							 	\State 	$B  \gets \{\}$. 
								\For{$v \in V$}:
									\If{$\RANK(B  + \{v\}) = \RANK(B ) + 1$}:
										\State $B  \gets B  + \{v\}.$
									\EndIf
								\EndFor
							\State \Return $B $. 					
							\EndProcedure
						\end{algorithmic}
				\end{varwidth}%
\end{algorithm}
Note that although described as a ``for-loop'', the above algorithm can be implemented in a single round of $n$ many $\RANK$ queries.
\end{proof}

\noindent
To obtain our reduction, what we need apart from this basis $B $ are two sets of {\em representatives}. A subset $T\subseteq V$ is a set of representative
if $|T\cap P_i| = 1$ for each $1\leq i\leq k$. The reduction will need {\em two} representative sets: $T_1 \subseteq B $ and $T_2 \cap B  = \emptyset$,
and the subroutine {\sc FindRepresentatives}($B$) will find this. Furthermore, it will also return a map $\phi:T_1 \to T_2$ where for each $e\in T_1$, $\phi(e)$ belongs
to the same part as $e$.
The algorithm does so in $O(n + k\log r)$ queries; in fact, only {\em independence oracle} queries suffice. 
This is slightly non-trivial and we described this in~\Cref{sec:findrep}. Let us now show how these representatives imply an $O(n)$-query algorithm to learn the partition $\calP$ and the $r_i$'s.

\begin{algorithm}[ht!]
	\caption{Using Representatives to Learn Partition}\label{alg:learnpartmatwreps}
	%
	\begin{varwidth}{\dimexpr\linewidth-2\fboxsep-2\fboxrule\relax}
		\begin{algorithmic}[1]
			\Procedure{LearnMatroidwithReps}{$V,\RANK, B, T_1, T_2,\phi: T_1\to T_2$}:
			\LineComment{Input: $n$ elements in $V$ with $\RANK$ query; basis $B$, set of representatives $T_1\subseteq B$, $T_2\cap B=\emptyset$, $\phi(t)$ is a friend of $t$.}
			\LineComment{Output: the partition $\calP$}
				\State For any subset $S\subseteq B$, define $\RANK_1(S) := \RANK(B - S + T_2) - \RANK(B-S)$.
				\State $\calP_1 \eq $ {\sc FindPartition}$(B, \RANK_1)$	\Comment{Takes $O(|B|)$ queries.} \label{alg:line:p1}
				\State For each $1\leq i\leq k$, $r_i \eq |B\cap P^{(1)}_i|$ where $\calP_1 = (P^{(1)}_1, \ldots, P^{(1)}_k)$.
				\State For any subset $S\subseteq V\setminus B$, define $\RANK_2(S) := \RANK(B + S - T_1) - \RANK(B-T_1)$.
				\State $\calP_2 \eq $ {\sc FindPartition}$(V\setminus B, \RANK_2)$	\Comment{Takes $O(|V\setminus B|)$ queries.} \label{alg:line:p2}
				\State Use $\phi$ to merge $\calP_1$ and $\calP_2$ into $\calP$: for $t\in T_1$, merge the part in $\calP_1$ containing $t_1$ with the part in $\calP_2$ containing $\phi(t)$.
				\State \Return $(\calP, \{r_i\}_{i=1}^k)$
			\EndProcedure
		\end{algorithmic}
	\end{varwidth}%
\end{algorithm}
\begin{claim}\label{clm:lpwr}
\Cref{alg:learnpartmatwreps} returns the correct partition $\calP$ and $r_i$'s making $O(n)$ many $\RANK$ queries.
\end{claim}
\begin{proof}
	The main idea is that the representatives allow us to simulate a simple partition matroid rank query on the basis and outside.
	More precisely, we claim that for any subset $S\subseteq B$,
	$
		\RANK_1(S) = \sum_{i=1}^k \min(|S\cap P_i|, 1).
	$
	If so, the correctness of~\Cref{alg:line:p1} follows from~\Cref{thm:partition}.
	Indeed, $\RANK(B - S + T_2) - \RANK(B-S)$ gives $+1$ for each part where $B-S$ loses at least one element to which the unique 
	element of $T_2$ contributes. Similarly, one argues that for any $S\subseteq V\setminus B$,
	$
	\RANK_2(S) = \sum_{i=1}^k \min(|S\cap P_i|, 1).
	$
	This is also for a similar reason; $B-T_1$ loses exactly one element from each part and so $\RANK(B+S-T_1) - \RANK(B-T_1)$ counts
	the parts that $S$ intersects at least once. See \Cref{fig2} for an illustration. 
%
\end{proof}

\begin{figure}[h]
    \centering
    \includegraphics[width=0.75\textwidth]{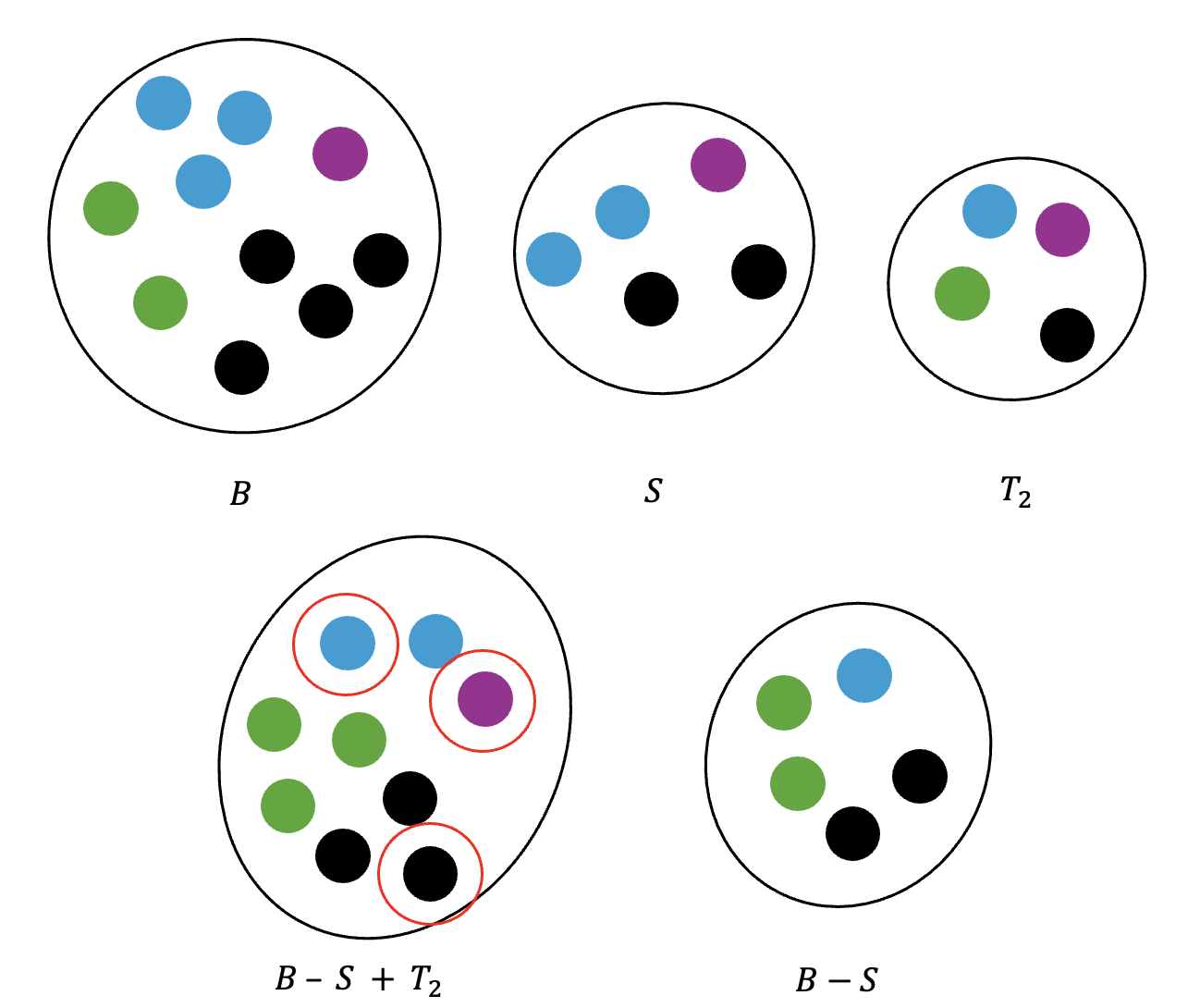}
    \caption{Illustration of how we simulate a simple partition matroid rank query inside of a basis with representatives outside. We have a basis with $b_i$s equal to $1$ (purple nodes), $2$ (green), $3$ (blue) and $4$ (black) respectively. $\RANK(B) = 10$, $\RANK(B- S) = 5$. Note $|B-S+T_2| = 10$, yet $\RANK(B-S+T_2) = 9$ because the number of green nodes is capped at 2. $\RANK(B-S+T_2) - \RANK(B-S) = 3$ simulate a simple partition matroid rank query for $S$. The circled nodes correspond to the 3 partitions included in $S$.}
    \label{fig2}
\end{figure}

\subsection{Finding Representatives via Binary Search}\label{sec:findrep}

We now describe a procedure which takes a basis $B$ of the partition matroid, and finds two subsets of representatives $T_1 \subseteq B$ and $T_2\cap B =\emptyset$.
The idea behind is a delicate binary search.
Fix some $e \in V \backslash B $ and without loss of generality,
say $e \in P_1$. We now show how to find one element in $B \cap P_1$ in $O(\log r)$ queries. The way to do it is by halving $B$ to $X_1 \sqcup X_2$ and keep the half with at least $P_1$ element in it as ``search half''. This can be checked by seeing whether $\RANKP(X_1 + \{e\}) = \RANKP(X_1)$: if so then $X_1$ contains all element of $P_1$ and this is the half we stick with; otherwise $X_2$ contains some $P_1$ element and takes precedence. The other half is now added to the new ``test half". We keep searching in our search set until it has exactly $1$ class $P_1$ element left. For instance, say $X_2$ is further divided to $X_{21}$ and $X_{22}$. The next query would be to check if $\RANKP((X_1 + \{e\}) + X_{21}) = \RANKP(X_1 + X_{21})$. If so, then $X_1 + X_{21}$ contains all class $P_1$ elements, and so will make $X_{21}$ our new ``search set''; otherwise, we continue on $X_{22}$. We describe the pseudocode in detail in~\Cref{alg:findreps}.

\begin{algorithm}[ht!]
	\caption{Finding Representatives}\label{alg:findreps}
	%
	\begin{varwidth}{\dimexpr\linewidth-2\fboxsep-2\fboxrule\relax}
		\begin{algorithmic}[1]
			\Procedure{FindRepresentatives}{$V,\RANK, B$}:
			\LineComment{Input: $n$ elements in $V$ with $\RANK$ query; basis $B$}
			\LineComment{Output: Sets of Representatives $T_1\subseteq B$, $T_2\cap B = \emptyset$ and map $\phi:T_1 \to T_2$.}
			\State $T_1,T_2 \eq \emptyset$.
			\For{$e\in V\setminus B$}:
				\If{$\RANK(B - T_1 + \{e\}) = \RANK(B - T_1)$}:\Comment{$e$ is an element with no friends in $T_1$ and $T_2$}: \label{alg:fr:line:if}
					\State $T_2 \eq T_2 + e$.
					\State $X\eq B$; $Y\eq \emptyset$.
					\While{$|X| > 1$}:
						\LineComment{Invariant: $X$ has at least one element in same part as $e$} \label{alg:fr:line:inv}
						\LineComment{Invariant: $X\cup Y = B$} \label{alg:fr:line:inv2}
						\State $(X_1, X_2) \eq$ arbitrary equipartition of $X$.
						\If{$\RANK(Y+X_1+e) = \RANK(Y+X_1)$}: \Comment{$X_2$ contains no friends of $e$} \label{alg:fr:line:if-1}
							\State $Y\eq Y + X_2$; ~~$X \eq X_1$.
						\Else: \Comment{$X_2$ contains at least one friend of $e$} \label{alg:fr:line:if-2}
							\State $Y \eq Y + X_1$; ~~$X \eq X_2$.
						\EndIf
					\EndWhile
					\LineComment{$X$ is a singleton element of $B$; let $X = \{x\}$}
					\State $T_1 \eq T_1 + x$; Set $\phi(x) = e$.
				\EndIf
			\EndFor
			\State \Return $(T_1, T_2, \phi)$.
			\EndProcedure
		\end{algorithmic}
	\end{varwidth}%
\end{algorithm}

\begin{claim}\label{clm:findreps}
	\Cref{alg:findreps} returns $(T_1, T_2, \phi)$ correctly and makes $O(n + k\log r)$ many $\RANK$ queries.
\end{claim}
\begin{proof}
	The proof is by induction: we claim that $T_1, T_2$ contains at most one element from each part and the size $|T_1| = |T_2|$ equals the number of 
	parts spanned by the elements seen by the outer for-loop. And furthermore, the $\phi$-relation is correct.
	This is obviously true before anything occurs, and consider the for-loop for and element $e$.
	Now suppose the if-statement in~\Cref{alg:fr:line:if} is {\em not} true; that is, say $\RANK(B - T_1  + e) > \RANK(B - T_1)$.
	This would mean that $e$ contains a friend in $T_1$; the only way the rank could increase is if $e$ filled the ``hole'' in the part which has exactly one element missing in $B-T_1$.
	We discard this $e$. On the other hand if the if-statement holds, then we will discover a new part in $B$ and thus by inductive hypothesis, in $V\setminus B$.
	We therefore add $e$ to $T_2$.
	
	Now consider the invariant in~\Cref{alg:fr:line:inv}. If that indeed holds true, then when the while-loop terminates, and it does so with $|X| = 1$, 
	the single element $x \in X$ must be in the same part as $e$. Thus, adding $x\in T_1$ and setting $\phi(x) = e$ is the correct thing to do.
	To see that the invariant in~\Cref{alg:fr:line:inv} holds, we include the invariant in~\Cref{alg:fr:line:inv2}. This is readily checked in both the ``then'' and ``else'' case
	of the forthcoming if-statement. If the~\Cref{alg:fr:line:if-1} holds true, then akin to the argument above, $Y+X_1$ contains all friends of $e$ in $B$.
	So, we focus our search on $X_1$ since it contains at least one friend of $e$ because, by invariant, $X$ contained at least one friend of $e$. So setting $X$ to $X_1$ keeps
	the invariant satisfied.
	On the other hand, if~\Cref{alg:fr:line:if-1} doesn't hold true, then $X_2$ must contain at least one friend of $e$. And so setting $X$ to $X_2$ keeps the invariant fulfilled. \smallskip
	
	\noindent
	To find the number of queries is simple. First notice that only~\Cref{alg:fr:line:if} and~\Cref{alg:fr:line:if-1} make any queries. And even then one of them is superfluous.
	More precisely, since $B-T_1$ and $Y+X_1$ are independent sets, their rank is $|B| - |T_1|$ and $|Y| + |X_1|$ respectively. 
	We make $n-r$ queries in line~\Cref{alg:fr:line:if-1}. Of these at most $k$ many satisfy the condition. Each of them leads to a binary-search style argument 
	which takes at most $\ceil{\log_2 r}$ many queries.
\end{proof}
\noindent
{\em Remark.
Note that~\Cref{alg:fr:line:if} and~\Cref{alg:fr:line:if-1} can be implemented using only independence-oracle queries since they are really asking, respectively, if 
$B-T_1 + e$ and $Y+X_1+e$ are independent or not; if the ranks are equal, they are not. This also implies an $O(n\log k)$ algorithm to learn the partition matroid
using only independence oracle as alluded to in the Introduction. 
	Let $|T_1| = |T_2| = k$.
Once we have the representative sets $T_1 \subseteq B$ and $T_2 \cap B = \emptyset$, for any element $e\notin V\setminus B$, we can use a binary-search style argument on $T_1$ to find
$e$'s friend among $T_1$ in $O(\log k)$ many independence oracle queries. More precisely, we halve $T_1$ into $(X,Y)$ and check if $B - X + e$ is independent or not. If it is, then $X$ contains $e$'s friend; otherwise, $Y$ does. Similarly, for any $e\in B$, we can find $e$'s friend in $T_2$ in $O(\log_2 k)$ many independence oracle queries.
}
\begin{algorithm}[ht!]
	\caption{Learning a Partition Matroid}\label{alg:learnpartmat}
	%
	\begin{varwidth}{\dimexpr\linewidth-2\fboxsep-2\fboxrule\relax}
		\begin{algorithmic}[1]
			\Procedure{LearnPartition}{$V,\RANK$}:
			\LineComment{Input: partition matroid on $n$ elements in $V$ with $\RANK$ query}
			\LineComment{Output: the partition $\calP$ and $r_i$'s}
			\State Learn a basis $B$ using {\sc FindBasis}($V,\RANK$) a la~\Cref{alg:basis}. \label{alg:final:1}  
			\State $(T_1,T_2,\phi) \eq $ {\sc FindRepresentatives} ($V,B,\RANK$) a la~\Cref{alg:findreps}. \label{alg:final:2}
			\State \Return $(\calP, \{r_i\}) \eq$ {\sc LearnMatroidWithReps}$(V,\RANK,B,T_1,T_2,\phi)$ a la~\Cref{alg:learnpartmatwreps}. \label{alg:final:3}
			\EndProcedure
		\end{algorithmic}
	\end{varwidth}%
\end{algorithm}

For completeness, we end the section by giving the pseudocode for the final algorithm in~\Cref{alg:learnpartmat}.
\Cref{lem:basis} establishes that \Cref{alg:final:1} makes $n$ $\RANK$ queries, \Cref{clm:findreps} establishes that \Cref{alg:final:2} makes
$n + k\log r$ $\RANK$ (in fact independence oracle) queries, and~\Cref{clm:lpwr} establishes that \Cref{alg:final:3} makes $O(n)$ $\RANK$ queries.
This completes the proof of~\Cref{thm:matroid}.

\section{Conclusion}

In this paper we looked at the question of learning a hidden partition using rank queries which given a subset tells how many different parts it hits.
We gave a simple but non-trivial, deterministic, and efficient algorithm which makes $O(n)$-rank queries. This is optimal up to constant factors. 
The main non-triviality arises in the use of techniques devised in coin-weighing algorithms a la~\cite{Cantor,Lindstrom}, and our work falls in a growing line 
of such results~\cite{GrebinskiK00,ChoiK10,BshoutyM10,Apers,ChakrabartyL23,Liao24a} which explores the use of these techniques to solve combinatorial search problems. 

The obvious question left open by our paper is whether there are $O(n)$ algorithms to learn general partition matroids especially when $k = \Theta(n)$.
We have not been able to directly port the coin-weighing techniques to solve this problem even in the case of $r_i = 2$ for all $i$. The main technical challenge
that the rank query, ultimately, is not a linear query and in~\Cref{sec:genpart} we could make it ``behave linear'' with the help of representatives.
Our algorithm to find representatives, however, didn't utilize the ``more information'' given by rank-queries over independence-oracle queries. Investigating this may lead to 
new algorithmic primitives. On the other hand, perhaps there is a $\omega(n)$ lower bound for this problem when $k = \Theta(n)$.	

\bibliographystyle{alpha}
\bibliography{main}
\end{document}